\documentclass{article}
\newcommand{\isdraft}{}
\usepackage[margin=3.5cm, top=2.5cm, bottom=3cm]{geometry}
\usepackage{amsmath,amsfonts,amsthm}
\usepackage{algorithmic}
\usepackage{algorithm}
\usepackage{array}
\usepackage[caption=false,font=normalsize,labelfont=sf,textfont=sf]{subfig}
\usepackage{textcomp}
\usepackage{stfloats}
\usepackage{url}
\usepackage{verbatim}
\usepackage{graphicx}
\usepackage{cite}
\usepackage
[colorlinks,
linkcolor=red!70!black,
citecolor=green!50!black,
urlcolor=magenta!60!black]
{hyperref}
\usepackage[utf8]{inputenc}
\usepackage{multirow}
\usepackage{authblk}

\usepackage{etoolbox,xcolor,lipsum,tikz,dsfont}

\newtheorem{theorem}{Theorem}
\newtheorem{proposition}{Proposition}

\newtheorem{lemma}{Lemma}

\theoremstyle{definition}
\newtheorem{definition}{Definition}
\newtheorem{protocol}{Protocol}

\renewcommand{\epsilon}{\varepsilon}
\newcommand{\Fold}{\ensuremath{\mathsf{Fold}}}

\newcommand{\Span}{\ensuremath{\mathrm{Span}}}
\newcommand{\RS}{\ensuremath{\mathsf{RS}}}

\newcommand{\Cut}{\ensuremath{\mathsf{Cut}}}
\newcommand{\Cay}{\ensuremath{\mathsf{Cay}}}
\newcommand{\loops}{\ensuremath{\mathsf{petals}}}
\newcommand{\desequilibre}{\mu}
\newcommand{\Card}[1]{\ensuremath{|{#1}|}}
\newcommand{\size}[1]{\ensuremath{|{#1}|}}
\newcommand{\V}{\calV}
\renewcommand{\P}{\calP}
\newcommand{\edges}[1]{\tilde{#1}}

\newcommand{\RIM}{RIM}

\newcommand{\decrease}{A}%{\mathsf{decrease}}
\newcommand{\reject}{R}%{\mathsf{reject}}
\newcommand{\invalidfold}{N}
\newcommand{\dinvalidfold}{N'}
\newcommand{\fprime}{f'}
\newcommand{\repetitionparameter}{m}
\newcommand{\numberofedgeschecked}{t}

\newcommand{\Hugo}[1]{{\footnotesize\color{orange!75!black}\textit{Hugo}~: #1}}

\newcommand{\vadisparaitre}[1]{#1}
\newcommand{\nouveau}[1]{{#1}}
\newcommand{\draftmath}[1]{\[#1\]}
\newcommand{\emptycommand}{}
\ifx\isdraft\emptycommand
  \renewcommand{\Hugo}[1]{}
  \renewcommand{\vadisparaitre}[1]{}
  \renewcommand{\draftmath}[1]{$#1$}
\fi

\docsvlist{A,B,C,D,E,F,G,H,I,J,K,L,M,N,O,P,Q,R,S,T,U,V,W,X,Y,Z,a,b,c,d,e,f,g,h,i,j,k,l,m,n,o,p,q,r,s,t,u,v,w,x,y,z}

\begin{document}

\title{Interactive Oracle Proofs of Proximity to Codes on Graphs}

% \author{IEEE Publication Technology,~\IEEEmembership{Staff,~IEEE,}% <-this % stops a space
% \thanks{Manuscript received October 8, 2023; revised December 8, 2023.}}

% \author{Hugo Delavenne, Tanguy Medevielle, Élina Roussel}
\author[1,2]{Hugo Delavenne}
\author[3,2]{Tanguy Medevielle}
\author[1,2]{Élina Roussel}
\affil[1]{LIX, École Polytechnique, Institut Polytechnique de Paris}
\affil[2]{INRIA}
\affil[3]{IRMAR, Université de Rennes}
\date{}
  % \IEEEauthorblockA{
  %   LIX, École Polytechnique, Institut Polytechnique de Paris\\
  %   INRIA\\
  %   hugo.delavenne@inria.fr,
  %   %\hspace{15mm}
  %   elinarousselcs@gmail.com
  %   % \\ \phantom{Ema} elina.roussel@inria.fr
  % }
  % \and

  % \IEEEauthorblockA{
  %   IRMAR, Université de Rennes\\
  %   INRIA\\
  %   tanguy.medevielle@univ-rennes.fr
  % }
  % }
%\author{}

% \markboth{Journal of \LaTeX\ Class Files,~Vol.~1, No.~2, December~2023}%

% \IEEEpubid{0000--0000~\copyright~2023 IEEE}

\maketitle

\begin{abstract}
  We design an Interactive Oracle Proof of Proximity (IOPP) for codes on graphs inspired by the FRI protocol. The soundness is significantly improved compared to the FRI, the complexity parameters are comparable, and there are no restrictions on the field used, enabling to consider new codes to design code-based SNARKs.

  % \textbf{Keywords:} IOPP, coding theory, proof of proximity, folding, codes on graphs, Cayley graphs, SNARK, verifiable computing
\end{abstract}

\section{Introduction}
\subsection{Scientific context}
Designing efficient Succinct Non-interactive ARguments of Knowledge (SNARKs) has become an important field in cryptography.
A SNARK is a cryptographic proof system that enables a computationally powerful prover to demonstrate the validity of a computational statement to a computationally weak verifier.
The SNARKs used in practice rely on the arithmetization of a computation to an algebraic problem, and on proving efficiently and interactively that the problem has a solution.
One of the main approaches relies on proximity tests to error-correcting codes as algebraic problem, and specifically to Reed-Solomon codes. Since they are built as evaluation of polynomials, they provide useful algebraic properties related to the arithmetization.
However Reed-Solomon codes are not locally-testable, meaning that testing proximity to the code requires access to a significant proportion of a word.
Interactive Oracle Proofs of Proximity (IOPP) \cite{BCS16,BBHR18a} were introduced to vercome this isue by enabling testing the proximity to a Reed-Solomon code while only reading a few coordinates.

An IOPP is an $r$-round interaction between a prover $\P$ and a verifier $\V$ in which $\P$ aims to convince $\V$ that, for a given word $f\in\bbF^n$, code $C\subseteq\bbF^n$ and parameter $\delta\in[0,1]$,
\begin{equation}
  \label{eq:protocol-claim}
  \Delta_H(f,C)\leq\delta\text,
\end{equation}
where $\Delta_H$ is the relative Hamming distance.
The Verifier receives oracles to the messages sent by $\P$. Then $\V$ only sends randomness in order to keep the protocol public coin, and to be able to apply a Fiat-Shamir transformation to turn the protocol into a non-interactive one \cite{FS87}, \cite[Section 6]{BCS16}.
The prover and verifier are modeled as $r+1$ algorithms $\P_0,...,\P_r$ and $\V_0,...,\V_r$ representing their behavior over rounds.
Each one takes as arguments the input of the protocol and the history of the interaction, and outputs the message to send to the other party.
Since the verifier does not access the whole input word, but only a chosen random part, there is a probability for it to accept at the end of the protocol even though (\ref{eq:protocol-claim}) is not satisfied. This probability is called the \emph{soundness}. On the other hand, when $f\in C$, the probability that the Verifier accepts is called the \emph{completeness}, and we require it to be $1$.

Codes on graphs were first introduced by Tanner \cite{Tan81}.
As in the celebrated construction by Sipser and Spielman \cite{SS96}, we consider codes whose coordinates are indexed by the edges of a graph.
Given a $n$-regular graph $\Gamma=(V,E)$ and a base code $C_0\subseteq\bbF^n$, the code on $\Gamma$ built on $C_0$ is the space of functions $f:E\to\bbF$ such that each ``local view'', i.e. the edges around a given vertex, is a codeword of $C_0$:
\[
  \forall v\in V,(f(e))_{e=(v,v')\in E}\in C_0\text.
\]
Expander graphs have been studied for their ability to provide codes with low locality, so one can test the proximity to them by only accessing a small proportion of the coordinates.
It led to the discovery of a family of codes with constant rate, constant minimal distance and constant locality \cite{DELLM21}.
However, these codes do not reach practical complexity parameters yet. Moreover their construction is very resctrictive and requires specific algebraic objects which are much more involved than simply polynomials like for Reed-Solomon codes, hence they are not easily suitable for arithmetization.
Unlike these cited constructions, we do not try to build a family of codes on graphs with a constant regularity, but we allow it to grow with the size of the graph.

\subsection{Techniques and results}

We adapt ideas from the FRI protocol \cite{BBHR18a} to create a new folding technique that reduces testing the proximity to a graph to testing the proximity to a graph with twice less vertices.
The specificity of our protocol lies in the preservation of the arity between the original and small graph.
A folding creates several ``loop edges'', referred to as ``\emph{petals}'', which ensures that the local views are preserved.
Repeating this folding technique produces a graph with a single vertex and only petals, referred to as the ``\emph{flower}''.
The protocol is thus called \emph{flowering}.
The graphs considered use a Reed-Solomon code as base code.

Our main improvement lies in the soundness of our protocol. It achieves a lower commit soundness factor compared to the FRI. Moreover, the soundness of the Flowering protocol remains valid up to the covering radius whereas it is only a conjecture for the FRI protocol \cite[Conjecture 2.3]{BGKS20}.
Furthermore, our protocol only require that the field size is larger than the logarithm of the length of the code to be able to define a Reed-Solomon code as base code, while the FRI requires much more structure on the field. \Hugo{Une référence sur le fait de chercher des bons corps ?}
\nouveau{However, the family of graphs on which we apply our protocol, defined in Section \ref{sec:cayley}, have a positive rate but an $o(1)$ minimum distance.}
We compare the parameters of the FRI protocol \cite{BBHR18a,BKS18,BCIKS23} for testing the proximity to a Reed-Solomon code of length $N$ and dimension $K$, with our Flowering protocol on the codes defined in Section \ref{sec:cayley}, of length $N$ and dimension $K$.
Table \ref{tab:comparison} presents the soundnesses and their domain of validity, as well as their complexity parameters.
For the FRI protocol, we rewrite the soundnesses from \cite{BKS18,BCIKS23} in a form that can be compared with this paper. We use \cite[Theorem 5.2]{BCIKS23} instead of \cite[Theorem 4.4]{BKS18} it in the proof of \cite[Theorem 7.2]{BKS18}.

\begin{table*}[th!]
  \centering
    \begin{tabular}{rc|c|c}
      & Protocol & FRI \cite{BBHR18a,BKS18,BCIKS23} & Flowering (this paper) \\ \hline
      \multirow{5}{*}{\rotatebox{90}{\small{}Complexities}} & Prover & $<8N$ & $<3N$ \\
      & Verifier & $<2\repetitionparameter\log K$ & $<4\repetitionparameter\numberofedgeschecked\log N$ \\
      & Query & $2\repetitionparameter\log K$ & $<2\repetitionparameter\numberofedgeschecked\log N$ \\
      & Length & $<N$ & $<N$ \\
      & Rounds & $\log K$ & $<\log N$ \\ \hline && \\[-3.5mm]
      \multirow{2}{*}{\rotatebox{90}{}} & Soundness & $\displaystyle\frac{K^2\log K}{(2\epsilon)^7\Card{\bbF}}+\left(1-\delta\right)^\repetitionparameter$ & $\displaystyle\frac{\log N}{\epsilon\Card{\bbF}}+\left(1-\frac{\numberofedgeschecked}{\log N}(\delta-\epsilon\log N)\right)^\repetitionparameter$ \\
      & Validity & $\delta< 1-\sqrt{\frac{K}{N}}-\epsilon$ & $\delta< 1-\frac{K}{N}$
  \end{tabular}
  \vspace*{2mm}
  \caption{
    Comparison of the soundnesses and complexity parameters,
    \nouveau{using the codes with constant rate and $o(1)$ minimum distance defined in Section \ref{sec:cayley} for the Flowering,}
    and a mixed soundness from \cite{BKS18,BCIKS23} for the FRI protocol, where $\repetitionparameter$ and $\numberofedgeschecked$ are repetition parameters, and $\epsilon$ is arbitrary.
  }
  \label{tab:comparison}
\end{table*}

\section{Definitions}

% In this work, we consider $n$-regular indexed multigraphs that can contain loop edges.
Fix a finite field $\bbF$ of cardinal greater than $n$ and $x_1,...,x_n\in\bbF$ pairwise distinct elements.
Denote the Reed-Solomon code evaluated on $\{x_1,...,x_n\}$ of dimension $k$ by $\RS[n,k]$.

\begin{definition}[Regular indexed multigraph (\RIM)]
  A \emph{$n$-regular indexed multigraph} $\Gamma=(V,E)$ is given by a set of vertices $V$ and a function $E:V\times[n]\to V$ such that for $v\in V$ and $\ell\in[n]$, the vertex $E(v,\ell)$ is the neighbor of $v$ through edge indexed $\ell$, and such that $E$ satisfies the well-definedness property: $\forall \ell, \forall v\in V, E(E(v,\ell),\ell)=v$.

  Denote $\edges{E}$ the quotient of $V\times[n]$ by the equivalence relation $\sim_E$ defined by $(v,\ell)\sim_E(v',\ell')$ iff $\ell=\ell'$, and either $E(v,\ell)=v'$ or $v=v'$.
  Denote $\loops(\Gamma):=\{(v,\ell)\in V\times[n]\mid E(v,\ell)=v\}$.
  % The number of edges of $\Gamma$ is denoted
  % \(
  % \size{E}:=\frac{n}{2}\Card{V}+\frac12\Card{\loops(\Gamma)}\text.
  % \)
\end{definition}

% \temp{We define codes on \RIM{} as usually for graphs, by indexing coordinates on the edges.}

\begin{definition}[Word on graph, code on graph]
  Let $\Gamma=(V,E)$ be a $n$-\RIM. A \emph{word $\edges{f}$ on $\Gamma$} is a function $\edges{E}\to\bbF$.
  % that satisfies the well-definedness property: $\forall \ell,\forall v\in V, f(v,\ell)=f(E(v,\ell),\ell)$.
  We denote $W(\Gamma,\bbF)$ the set of functions $f:V\times[n]\to\bbF$ such that for any $v,v',\ell$, if $(v,\ell)\sim_E(v',\ell)$ then $f(v,\ell)=f(v',\ell)$.
  A word on $\Gamma$ can equivalently be viewed as a function in $W(\Gamma,\bbF)$. Since this formalism will be more convenient here, we will use it instead.

  For $f\in W(\Gamma,\bbF)$ and $v\in V$, we denote by $f(v,\cdot)$ the vector $(f(v,1),...,f(v,n))$, and we denote $f(v,X)$ the degree $<n$ polynomial such that $f(v,x_i)=f(v,i)$ for $i=1,...,n$.

  Define the code $\calC[\Gamma,k]$ on $\Gamma$ as $\{f\in W(\Gamma,\bbF)\mid \forall v\in V,f(v,\cdot)\in \RS[n,k]\}$.
\end{definition}

% We will denote the length, dimension and minimal distance of codes on graphs with capital letters, and with small letters for the code $\RS[n,k]$ on which they are defined.

\begin{definition}[Graph isomorphism]
  Let $\Gamma=(V,E)$ and $\Gamma'=(V',E')$ be $n$-\RIM{}. An \emph{isomorphism} between $\Gamma$ and $\Gamma'$ is a bijection $\varphi:V\to V'$ such that $\forall v\in V, \forall \ell\in [n], \varphi(E(v,\ell))=E'(\varphi(v),\ell)$.
\end{definition}

\begin{definition}[Cut-graph, cut-word]
  Let $\Gamma=(V,E)$ be a $n$-\RIM{}. For $V'\subseteq V$, $\Cut[\Gamma,V']$ is defined as the $n$-\RIM{} $(V',E')$ where
  \[E':(v,\ell)\mapsto
    \begin{cases}
      E(v,\ell)&\text{if }E(v,\ell)\in V'\\
      v&\text{otherwise.}
    \end{cases}\]

  Let $f\in W(\Gamma,\bbF)$. For $V'\subseteq V$, we define the word $\Cut[f,V']$ on $\Cut[\Gamma,V']$ as the restriction of $f$ to $\Cut[\Gamma,V']$: $\forall v\in V',\forall\ell\in[n],\Cut[f,V'](v,\ell):=f(v,\ell)$.

  % Let $C:=\calC[\Gamma,k]$. For $V'\subseteq V$, we define the code $\Cut[C,V']$ on $\Cut[\Gamma,V']$ as $\{\Cut[f,V']\mid f\in C\}$. \Hugo{J'ai l'impression que c'est juste $\calC[\Cut[\Gamma,V'],k]$, donc c'est peut-être pas nécessaire de le mettre.}
\end{definition}

\begin{definition}[Flowering cut]
  Let $\Gamma=(V,E)$ be a $n$-\RIM{}.
  Let $V',V''\subseteq V$ be a partition of $V$.
  If there exists an isomorphism $\varphi:\Cut[\Gamma,V']\to\Cut[\Gamma,V'']$, then $F=(V',\varphi)$ is said to be a \emph{flowering cut}.

  Denote $\pi_\varphi:V\to V'$ the projection such that $\pi_\varphi(v)=v$ if $v\in V'$ and $\varphi^{-1}(v)$ otherwise.
\end{definition}

We define a folding notion, similar to \cite{BLNR20}.

\begin{definition}[Folding]
  Let $\Gamma=(V,E)$ be a regular well-defined \RIM{}. Let $f\in W(\Gamma,\bbF)$. Let $F=(V',\varphi)$ be a flowering cut and denote $V'':=V\setminus V'$. Denote $f':=\Cut[f,V']$ and $f'':=\Cut[f,V'']$. We define the \emph{folding} of $f$ on the cut $F$ by $\alpha\in\bbF$ as the following word of $W(\Cut[\Gamma,V'],\bbF)$
  \[\Fold_F[f,\alpha]:(v,\ell)\mapsto f'(v,\ell)+\alpha f''(\varphi(v),\ell)\text.\]
\end{definition}

When it is clear from context, we denote $\Fold$ that operator.
% Following this definition, we are naturally interested in codes that can be successively folded.

\begin{definition}[Blossoming graph sequence]
  \label{definition:blossoming-graph-sequence}
  A sequence of $n$-\RIM{} $(\Gamma_0=(V_0,E_0),...,\Gamma_r=(V_r,E_r))$ is said to be \emph{blossoming} if $\Gamma_r$ has exactly $1$ vertex, and for any $i=1,...,r$, there exists a flowering cut $F_i=(V_{i},\varphi_i)$ such that $\Gamma_i=\Cut[\Gamma_{i-1},V_i]$.
\end{definition}

We introduce a distance called vertex distance, more suitable for the local views. It is coarser than the Hamming distance.

\begin{definition}[Vertex distance, Hamming distance]
  Let ${\Gamma=(V,E)}$ be a $n$-\RIM{}. Let $f,f'\in W(\Gamma,\bbF)$. We define the relative vertex distance between $f$ and $f'$, denoted $\Delta_V$, by
  \[\Delta_V(f,f'):=\frac{1}{\Card{V}}\Card{\{v\in V\mid f(v,\cdot)\neq f'(v,\cdot)\}}\text.\]
  % Let $\mathrm{diff}(f,f'):=\{(v,\ell)\in V\times[n]\mid f(v,\ell)\neq f'(v,\ell)\}$.
  We reformulate the relative Hamming distance between $f$ and $f'$, denoted $\Delta_H$, by
  \[
    \Delta_H(f,f'):=\frac{1}{\Card{\edges{E}}}\Card{\left\{\overline{(v,\ell)}\in\edges{E}\mid f(v,\ell)\neq f'(v,\ell)\right\}}\text.
  \]
\end{definition}

\begin{proposition}
  \label{proposition:comparison-distances}
  Let $\Gamma=(V,E)$ be a $n$-\RIM{}. Let $f,f'\in W(\Gamma,\bbF)$.
  For $v\in V$ and $\ell\in[n]$, let $\size{\overline{(v,\ell)}}$ be the cardinal of the equivalence class of $(v,\ell)$ by $\sim_E$,
  let $m:=\max_{v\in V}\sum_{\ell\in[n]}\frac{1}{\size{\overline{(v,\ell)}}}$.
  Then $\Delta_V(f,f')\geq\frac{\Card{\edges{E}}}{m\Card{V}}\Delta_H(f,f')$.
\end{proposition}

\begin{proof}
  For $v\in V$ and $\ell\in[n]$, let $d(v):=1$ if $f(v,\cdot)\neq f'(v,\cdot)$ and $0$ otherwise,
  let $d(v,\ell):=1$ if $f(v,\ell)\neq f'(v,\ell)$ and $0$ otherwise.
  % let $\size{v}:=\sum_{\ell\in[n]}\frac{1}{\size{(v,\ell)}}$.
  Since $d(v)m\geq\sum_{\ell\in[n]}\frac{d(v)}{\size{\overline{(v,\ell)}}}\geq\sum_{\ell\in[n]}\frac{d(v,\ell)}{\size{\overline{(v,\ell)}}}$,
  we have
  % $d(v)\geq\frac{1}{m}\sum_{\ell\in[n]}\frac{d(v,\ell)}{\size{(v,\ell)}}$. Thus
  $\Delta_V(f,f')=\frac{1}{\Card{V}}\sum_{v\in V}d(v)\geq\frac{1}{m\Card{V}}\sum_{v,\ell}\frac{d(v,\ell)}{\size{\overline{(v,\ell)}}}$.
  Moreover, $\Delta_H(f,f')=\frac{1}{\Card{\edges{E}}}\sum_{v,\ell}\frac{d(v,\ell)}{\size{\overline{(v,\ell)}}}=\frac{m\Card{V}}{\Card{\edges{E}}}\frac{1}{m\Card{V}}\sum_{v,\ell}\frac{d(v,\ell)}{\size{\overline{(v,\ell)}}}$, which gives the result.
\end{proof}

Denote $\desequilibre(\Gamma)$ the ratio $\frac{\Card{\edges{E}}}{m\Card{V}}$.
As a corollary of Proposition \ref{proposition:comparison-distances}, if each vertex of $\Gamma$ has the same amount of loops, then $\desequilibre(\Gamma)=1$ and thus $\Delta_V(f,f')\geq\Delta_H(f,f')$.
In Section \ref{sec:cayley}, this will be satisfied.

\section{General properties}

We adapt the lower bound on the dimension from \cite[Theorem 1]{Tan81} to our construction.

\begin{proposition}[Lower bound on the dimension]
  Let $\Gamma=(V,E)$ be a $n$-\RIM{}.
  Then $K:=\dim \calC[\Gamma,k]\geq (k-n/2)\Card{V}+\Card{\loops(\Gamma)}/2$.
\end{proposition}

\begin{proof}
  By aggregating the $\Card{V}$ parity check matrices of all the vertices for the code $\RS[n,k]$, one obtains a parity check matrix $H$ for $\calC[\Gamma,k]$, with $\Card{\edges{E}}$ columns and $(n-k)\Card{V}$ rows. Thus $\dim \calC[\Gamma,k]=\dim\ker H\geq \Card{\edges{E}}-(n-k)\Card{V}=(k-n/2)\Card{V}+\Card{\loops(\Gamma)}/2$.
\end{proof}

Proposition \ref{proposition:commit-soundness} is the graph analog of \cite[Theorem 4.4]{BKS18}.

\begin{proposition}[Commit soundness]
  \label{proposition:commit-soundness}
  Let $\epsilon>0$. Let $\Gamma=(V,E)$ be a $n$-\RIM{}.
  Let $F=(V',\varphi)$ be a flowering cut.
  Denote $C:=\calC[\Gamma,k]$, $\Gamma'=\Cut[\Gamma,V']=(V',E')$ and $C':=\calC[\Cut[\Gamma,V'],k]$.
  Let $f\in W(\Gamma,\bbF)$. Then
  \begin{equation*}\label{eq:commit-soundness}
    \underset{\alpha\in\bbF}{\Pr}\left[\Delta_V(\Fold_F[f,\alpha],C')<\Delta_V(f,C)-\epsilon\right]\leq\frac{1}{\epsilon\Card{\bbF}}\text.
  \end{equation*}
\end{proposition}

\begin{proof}
  Denote $\delta:=\Delta_V(f,C)$ and assume that $\delta>0$. Let $T:=\{v\in V\mid f(v,\cdot)\notin \RS[n,k]\}$, and $T':=\pi_\varphi(T)$ be the vertices of $\Gamma'$ whose $\Fold$ is built from at least one vertex of $T$.
  We have that $\forall v'\in T', \Card{\pi_\varphi^{-1}(v')\cap T}\leq 2$, hence
  \begin{equation}\label{eq:TgeqtwoTprime}
    \Card{T}=\sum_{v'\in T'}\Card{\pi_\varphi^{-1}(v')\cap T}\leq 2\Card{T'}\text.
  \end{equation}
  Furthermore, by definition of the vertex distance, $\Card{T}=\delta\Card{V}$, since $\Card{V}=2\Card{V'}$ and by (\ref{eq:TgeqtwoTprime}), we have
  \begin{equation}\label{eq:inequality-T-Tprime}
    \Card{T'}\geq\frac{\Card{T}}{2}=\frac{\delta\Card{V}}{2}=\delta\Card{V'}\text.
  \end{equation}
  For $\alpha\in\bbF$, denote $V_\alpha:=\{v'\in T'\mid \Fold[f,\alpha](v',\cdot)\in\RS[n,k]\}$.
  Then
  \begin{align}
    &\Pr(\Delta_V(\Fold[f,\alpha],C')<\delta-\epsilon) \nonumber \\
    &= \Pr\left(\Card{\{v'\in V'\mid \Fold[f,\alpha](v',\cdot)\notin \RS[n,k]\}}<(\delta-\epsilon)\Card{V'}\right) \nonumber \\
    &= \Pr\left(\Card{\{v'\in T'\mid \Fold[f,\alpha](v',\cdot)\notin \RS[n,k]\}}<(\delta-\epsilon)\Card{V'}\right) \nonumber \\
    &= \Pr\left(\Card{V_\alpha}>\Card{T'}-(\delta-\epsilon)\Card{V'}\right) \nonumber \\
    &\leq \Pr\left(\Card{V_\alpha}>\epsilon\Card{V'}\right)\text, \label{eq:inequality-validnodes}
  \end{align}
  where (\ref{eq:inequality-validnodes}) is obtained by (\ref{eq:inequality-T-Tprime}).
  I.e., with $A:=\{\alpha\in\bbF\mid \Card{V_\alpha}>\epsilon\Card{V'}\}$,
  \begin{equation}\label{eq:inequality-fold-getting-closer-leq-A}
    \underset{\alpha}{\Pr}(\Delta_V(\Fold[f,\alpha],C')<\delta-\epsilon)\leq\frac{\Card{A}}{\Card{\bbF}}\text.
  \end{equation}

  We now provide a bound on $\Card{A}$.
  Let $v'\in T'$ and denote $A_{v'}:=\{\alpha\in\bbF\mid \Fold[f,\alpha](v',\cdot)\in \RS[n,k]\}$.
  Denote $\sum_{i=0}^da_iX^i:=\Cut[f,V'](v',X)$ and $\sum_{i=0}^db_iX^i:=\Cut[f,V\setminus V'](v',X)$, where $d$ is the maximum degree of $\Cut[f,V'](v',X)$ and $\Cut[f,V\setminus V'](v',X)$.
  Since $f\notin C$, $d\geq k$, and because
  \draftmath{\Fold[f,\alpha](v',X)=\sum_{i=0}^d(a_i+\alpha b_i)X^i\text,}
  there is at most one value $\alpha$ such that $\deg\Fold[f,\alpha](v',X)<d$, hence
  \begin{equation}\label{eq:cardAv}
    \Card{A_{v'}}\leq 1\text.
  \end{equation}

  On the one hand, by definition of $A$,
  \draftmath{
    \sum_{\alpha\in A}\sum_{v'\in T'}\mathds{1}_{A_{v'}}=\sum_{\alpha\in A}\Card{V_\alpha}>\epsilon\Card{A}\Card{V'}\text,
  }
  and on the other hand, by (\ref{eq:cardAv}),
  \draftmath{
    \sum_{v'\in T'}\sum_{\alpha\in A}\mathds{1}_{A_{v'}}= \sum_{v'\in T'}\Card{A_{v'}}\leq\Card{T'}\leq\Card{V'}\text.
  }

  Thus $\Card{A}\leq\frac{1}{\epsilon}$, and with (\ref{eq:inequality-fold-getting-closer-leq-A}) we obtain the result.
\end{proof}

\noindent
\nouveau{We achieve a better soundness than~\cite{BKS18} because $\Fold[f,\alpha]$ gets closer to $C'$ only if a local-view becomes an RS codeword, i.e. a linear combination of non-codewords is a codeword.}

\section{Flowering protocol}

Protocol \ref{protocol:flowering} is an IOPP inspired from the FRI protocol \cite{BBHR18a}. The Prover has access to the word $f$ on $\Gamma$, and the Verifier has oracle access to $f$. The Prover aims to convince the Verifier that $\Delta_H(f,\calC[\Gamma,k])\leq\delta$. For this, the Prover will successively reduce the problem to testing the proximity to smaller codes.

\subsection{Protocol and properties}

Let $(\Gamma_0=(V_0,E_0),...,\Gamma_r=(V_r,E_r))$ be a blossoming $n$-\RIM{} sequence on the flowering cuts $F_1=(V_1,\varphi_1),...,F_r=(V_r,\varphi_r)$. For $i=0,...,r$, let $C_i:=\calC[\Gamma_i,k]$. Note that $C_r=\RS[n,k]$ is the code on the singleton \RIM{} with $n$ loop vertices, which we call a flower.

\begin{protocol}[Flowering protocol]
  \label{protocol:flowering}
  The flowering protocol is composed of two phases: the commit phase and the query phase. There are two complexity parameters: the number $\repetitionparameter$ of repetitions of the query phase and the number $\numberofedgeschecked$ of edges that are checked.

  \textsc{Commit phase:}
  For $i$ from $1$ to $r$, the $\V_i$ sends $\alpha_{i-1}\overset{\$}{\leftarrow}\bbF$ to $\P$ and $\P_i$ gives to $\V$ oracle access to a word $f_i\in W(\Gamma_i,\bbF)$.

  \textsc{Query phase:}
  For $j\in[\repetitionparameter]$, $\calV_r$ picks $v_{0,j}\overset{\$}{\leftarrow}V_0$ and a random set $I_j\subseteq[n]$ of size $\numberofedgeschecked$.
  For $i\in[r]$, $\V_r$ computes $v_i:=\pi_{\varphi_{i-1}}(v_{i-1})$, and checks that
  \begin{equation*}
    \label{eq:query-test}
    \forall\ell\in I_j,\Fold[f_{i-1},\alpha_{i-1}](v_{i,j},\ell)=f_i(v_{i,j},\ell)
  \end{equation*}
  by making $2\numberofedgeschecked$ queries if $i-1=0$, or $\numberofedgeschecked$ if $i\geq 2$, to $f_{i-1}$, and $\numberofedgeschecked$ queries to $f_i$.
  Finally with $v_r$ the only vertex of $\Gamma_r$, $\calV_r$ checks that
  \[
    f_r(v_r,\cdot)\in\RS[n,k]\text.
  \]
  The Verifier accepts only if all checks pass.
\end{protocol}

\begin{theorem}[Complexity properties of the protocol]
  Protocol \ref{protocol:flowering} has the following complexity properties
  \begin{itemize}
  \item Prover complexity: $3\sum_{i=1}^r\Card{\edges{E_i}}<3n\Card{V_0}$;
  \item Verifier complexity: $4r\repetitionparameter\numberofedgeschecked$;
  \item Query complexity: $(2r+1)\repetitionparameter\numberofedgeschecked+n$;
  \item Round complexity: $r$;
  \item Randomness complexity: $r$ fields elements, $\repetitionparameter$ nodes and $\repetitionparameter$ subsets of $[n]$;
  \item Proof length: $\sum_{i=1}^r\Card{\edges{E_i}}<n\Card{V_0}$.
  \end{itemize}
\end{theorem}

\begin{theorem}
  \label{thm:protocol-parameters}
  Let $(\Gamma_0,...,\Gamma_r)$ be a blossoming $n$-\RIM{}.
  The following properties hold when running Protocol \ref{protocol:flowering} on a word $f\in W(\Gamma_0,\bbF)$ with $\repetitionparameter$ repetitions of the query phase by checking $\numberofedgeschecked$ edges, where the probabilities are taken over the Verifier's internal randomness.

  \begin{enumerate}
  \item Completeness: if $f\in C_0$ then there exists a prover $\P$ such that $\V$ accepts with probability $1$.
  \item Soundness: for any prover $\P$, $\V$ accepts with probability at most
    \[\underset{\epsilon>0}{\min}\left(\frac{r}{\epsilon\Card{\bbF}}+\left(1-\frac{\numberofedgeschecked}{n}\big(\desequilibre(\Gamma)\Delta_H(f,\calC[\Gamma_0,k])-r\epsilon\big)\right)^\repetitionparameter\right)\text.\]
  \end{enumerate}
\end{theorem}

The completeness property is straightforward since the Prover can send $f_i=\Fold[f_{i-1},\alpha_{i-1}]$ for $i=1,...,r$ to make the Verifier accept with probability $1$.

\subsection{Proof of soundness}

We prove the soundness of the theorem, stated in vertex distance, in Proposition \ref{proposition:soundness-with-vertex-distance}, using the same strategy as for \cite[Theorem 7.2]{BKS18}.

\begin{lemma}
  \label{lemma:disjoint-events}
  Using the notations of Protocol \ref{protocol:flowering}, let $j\in[\repetitionparameter]$ be fixed.
  Let $\invalidfold_{i,j}$ denote the event ``$\Fold[f_{i-1},\alpha_{i-1}](v_{i,j},\cdot)\neq f_i(v_{i,j},\cdot)$''.
  For $(\fprime_0,...,\fprime_r)\in \prod_{i=0}^rW(\Gamma_i,\bbF)$, let $\dinvalidfold_{i,j}$ denote the event ``$\Fold[\fprime_{i-1},\alpha_{i-1}](v_{i,j},\cdot)\neq \fprime_i(v_{i,j},\cdot)$''.
  There exists $(\fprime_0,...,\fprime_r)\in \prod_{i=0}^rW(\Gamma_i,\bbF)$ such that
  $f'_r=f_r$, and the events $\dinvalidfold_{1,j},...,\dinvalidfold_{r,j}$ are disjoint and $\bigsqcup_{i=1}^r\dinvalidfold_{i,j}\subseteq\bigcup_{i=1}^r\invalidfold_{i,j}$.
\end{lemma}

\begin{proof}
  Define recursively $\tilde{f}_i$ for $i=0,..,r$ by $\tilde{f}_0:=f_0$ and for $i>0$, $\tilde{f}_i:=\Fold[\tilde{f}_{i-1},\alpha_{i-1}]$.
  We define the $\fprime_i$ as follows.
  For $i=0$, $\fprime_0=f_0$.
  Let $v_0\in V_0$.
  Denote $(v_1,...,v_r)$ the sequence such that for $i\in[r]$, $v_{i}=\pi_{\varphi_{i-1}}(v_{i-1})$.
  Denote $i(v_0):=\max(\{i\in[r]\mid f_{i}(v_{i},\cdot)\neq\Fold[f_{i-1},\alpha_{i-1}](v_{i},\cdot)\}\cup\{0\})$.
  Let $\ell\in[n]$.
  For $i<i(v_0)$ define $\fprime_i(v_i,\ell):=\tilde{f}_i(v_i,\ell)$, and for $i\geq i(v_0)$ define $\fprime_i(v_i,\ell):=f_i(v_i,\ell)$.
  Let $i_0:=i(v_{0,j})$.
  If $i_0=0$ then for any $i\in[r]$, $\dinvalidfold_{i,j}$ does not hold.
  If $i_0>0$ then by construction, for $i\in[r]\setminus\{i_0\}$, $\dinvalidfold_{i,j}$ does not hold.
  Thus the $(\dinvalidfold_{i,j})_{i\in[r]}$ are disjoint.
  Moreover, if the event $\bigsqcup_{i=1}^r\dinvalidfold_{i,j}$ holds,
  then $i_0>0$
  i.e. the event $\bigcup_{i=1}^r\invalidfold_{i,j}$ holds.
\end{proof}

\begin{proposition}[Query soundness]
  \label{proposition:soundness-with-vertex-distance}
  Let $\epsilon>0$ and $f_0\in W(\Gamma_0,\bbF)$. After running Protocol \ref{protocol:flowering} with $\repetitionparameter$ repetitions of the query phase by checking $\numberofedgeschecked$ edges, the Verifier accepts with probability at most
  \[\frac{r}{\epsilon\Card{\bbF}}+\left(1-\frac{\numberofedgeschecked}{n}(\Delta_V(f_0,\calC[\Gamma_0,k])-r\epsilon)\right)^\repetitionparameter\text,\]
  where the probability is taken over the its internal randomness.
\end{proposition}

\begin{proof}
  % During the protocol, \calV{} picks independently at random $\alpha_0,...,\alpha_{r-1}\in\bbF$ and $v_{1,1},...,v_{1,\repetitionparameter}\in V_0$.

  If $f_r\notin C_r$ then the Verifier rejects with probability $1$. Therefore in the following we assume that
  \begin{equation}
    \label{eq:fr-in-code}
    \Delta_V(f_r,C_r)=0\text.
  \end{equation}
  Let $\fprime_0,...,\fprime_r$ be given by Lemma \ref{lemma:disjoint-events}.
  For $i\in[r]$ and $j\in[\repetitionparameter]$ denote
  $\reject_{i,j}$ the event ``$\exists \ell\in I_j$ such that $\Fold[f_{i-1},\alpha_{i-1}](v_{i,j},\ell)\neq f_i(v_{i,j},\ell)$''.
  Denote $\decrease$ the event ``$\forall i\in[r], \Delta_V(\Fold[\fprime_{i-1},\alpha_{i-1}],C_i)\geq\Delta_V(\fprime_{i-1},C_{i-1})-\epsilon$''.
  Then the event ``\calV{} accepts'' is $\bigcap_{j=1}^\repetitionparameter\bigcap_{i=1}^r\overline{\reject_{i,j}}$.

  By the law of total probability, the probability that \calV{} accepts is at most $\Pr\left(\overline{\decrease}\right)+\Pr\left(\bigcap_{i,j}\overline{\reject_{i,j}}\mid \decrease\right)$.
  By Proposition \ref{proposition:commit-soundness}, $\Pr\left(\overline{\decrease}\right)\leq\sum_{i=1}^r\Pr\big[\Delta_V(\Fold[\fprime_{i-1},\alpha_{i-1}],C_i)<\Delta_V(\fprime_{i-1},C_{i-1})-\epsilon\big]\leq\frac{r}{\epsilon\Card{\bbF}}$.
  By independence of the repetitions of the query phase,
  \begin{equation}
    \label{eq:independence-repetition}
    \Pr\left(\bigcap_{i,j}\overline{\reject_{i,j}}\mid \decrease\right)=\prod_{j=1}^\repetitionparameter\left(1-\Pr\left(\bigcup_{i\in[r]}\reject_{i,j}\mid \decrease\right)\right)\text.
  \end{equation}
  Since all $\Pr(\bigcup_{i=1}^r\reject_{i,j}\mid \decrease)$ for $j\in[\repetitionparameter]$ are equal, we consider the case $j=1$.
  Take the notations $\invalidfold_{i,1}$ and $\dinvalidfold_{i,1}$ of Lemma \ref{lemma:disjoint-events}.
  Then $\Pr\left(\bigcup_{i\in[r]}\reject_{i,1}\mid \decrease\cap\bigcup_{i\in[r]}\invalidfold_{i,1}\right)\geq\frac{\numberofedgeschecked}{n}$
  and therefore,% by the total probabilities,
  \begin{align*}%\label{eq:inequality-probability-not-all-edges}
    \Pr\left(\bigcup_{i\in[r]}\reject_{i,1}\mid\decrease\right)
    &\geq\frac{\numberofedgeschecked}{n}\Pr\left(\bigcup_{i\in[r]}\invalidfold_{i,1}\mid\decrease\right)\\
    &\geq\frac{\numberofedgeschecked}{n}\Pr\left(\bigsqcup_{i\in[r]}\dinvalidfold_{i,1}\mid\decrease\right)\text.
  \end{align*}
  Hence by (\ref{eq:independence-repetition}),% and since $\bigsqcup_{i=1}^r\dinvalidfold_{i,j}\subseteq\bigcup_{i=1}^r\invalidfold_{i,j}$,
  \begin{equation}
    \label{eq:query-soundness-in-disjoint-events}
    \Pr\left(\bigcap_{i,j}\overline{\reject_{i,j}}\mid \decrease\right)\leq\left(1-\frac{t}{n}\sum_{i=1}^r\Pr(\dinvalidfold_{i,1}\mid A)\right)^\repetitionparameter\text.
  \end{equation}
  Assuming $\decrease$ holds, by denoting $\delta_i:=\Delta_V(\fprime_i,C_i)$, by the triangle inequality we have that
  \begin{align*}
    \delta_i
    &\geq\Delta_V(\Fold[\fprime_{i-1},\alpha_{i-1}],C_i)-\Delta_V(\fprime_i,\Fold[\fprime_{i-1},\alpha_{i-1}])\\
    &\geq \delta_{i-1}-\epsilon-\Pr(\dinvalidfold_{i,1})\text.
  \end{align*}
  Thus $\Pr(\dinvalidfold_{i,1}\mid \decrease)\geq\delta_{i-1}-\delta_i-\epsilon$ and by telescoping,
  \begin{equation}
    \label{eq:telescopic-sum}
    \sum_{i=1}^r\Pr(\dinvalidfold_{i,1}\mid A)\geq \delta_0-\delta_r-r\epsilon\text.
  \end{equation}
  By construction, $\fprime_r=f_r$, therefore by (\ref{eq:fr-in-code}), (\ref{eq:query-soundness-in-disjoint-events}) and (\ref{eq:telescopic-sum}) we get the result.
\end{proof}

\noindent
Theorem \ref{thm:protocol-parameters} is then a corollary of Propositions \ref{proposition:comparison-distances} and \ref{proposition:soundness-with-vertex-distance}.

\section{Cayley multigraph over $(\bbF_2^r,+)$}
\label{sec:cayley}

We instantiate the Protocol \ref{protocol:flowering} by defining a blossoming graph sequence built from Cayley graphs \cite{Cay78} over the additive group $\bbF_2^r$, and we prove a bound for their minimal distance. \Hugo{Ça fait sens de citer ça ??}
% When $G$ is a group and $S\subseteq G$, we denote $\langleS\rangle$ the subgroup generated by $S$.

\begin{definition}[Cayley \RIM]
  Let $G$ be a finite group and $S=\{s_1,...,s_n\}\subseteq G$ be symmetric. We define the $n$-\RIM{} $\Cay(G,S)=(V,E)$ by $V=G$ and $E:(v,\ell)\mapsto v\cdot s_\ell$.
\end{definition}

\begin{definition}[Blossoming Cayley multigraph sequence]
  Let $G=\bbF_2^r$ and $S\subseteq G$.
  We define the \emph{blossoming graph sequence} $\Gamma_0,...,\Gamma_r$ as follows. $\Gamma_0:=\Cay(G,S)$ and for $i>0$, define $V_i:=\{0\}^i\times\bbF_2^{r-i}$ and
  \[
    \varphi_i:\underset{i-1\text{ zeroes}}{(\underbrace{0,...,0},}1,g_{i+1},...,g_r)\mapsto(\underset{i\text{ zeroes}}{\underbrace{0,...,0}},g_{i+1},...,g_r)
  \]
  and $\Gamma_i:=\Cut[\Gamma_{i-1},V_i]$.
\end{definition}

In the following, we assume that $r\leq n$ and there exists $d\leq r+1$ such that there exists a binary code of parameters $[n,n-r,d]_2$. Then by taking $S$ the set of columns of a parity check matrix of that code, we obtain Lemma \ref{lemma:generating-set-independent-subsets}.
This construction is called a coset-graph \cite{BZ22}.

\begin{lemma}
  \label{lemma:generating-set-independent-subsets}
  There exists $S\subseteq\bbF_2^r$ such that $\Card{S}=n$, $\Span(S)=\bbF_2^r$ and any subset of $d-1$ vectors of $S$ are linearly independent.
\end{lemma}

% \begin{proof}
%   Let $C$ the code given by (\ref{eq:code-minimal-distance-d}).
%   Let $H\subseteq\bbF_2^{r\times n}$ be a parity check matrix for $C$. %\Hugo{Il faudrait quand même justifier que ça existe mais Daniel a l'air convaincu.}
%   Let $S\subseteq\bbF_2^r$ be the set of columns of $H$.
%   Since $C$ has length $n$, $\Card{S}=n$. Since $C$ has dimension $n-r$, $\Span(S)=\bbF_2^r$. Since $C$ has minimal distance $d$, any subset of $d-1$ vectors of $S$ are linearly independent.
% \end{proof}

\noindent
With $G=(\bbF_2^r,+)$, $\Card{S}=n$ and $k\leq n$, remark that the length of $\calC[\Cay(G,S),k]$ is thus $n2^{r-1}$ and its rate is at least $\frac{2k}{n}-1$.

\begin{proposition}[Lower bound on the minimal distance]
  \label{proposition:lower-bound-minimal-distance}
  Assume that $n-k+1=d-1$.
  If $S$ is given by Lemma \ref{lemma:generating-set-independent-subsets}, with $\Gamma=(V,E):=\Cay(\bbF_2^r,S)$, then
  \begin{equation*}
    \label{eq:lower-bound-minimal-distance}
    \Delta_H(\calC[\Gamma,k])\geq 2^{d-r-2}\!\left(1-\frac{k-1}{n}\right)\text.
  \end{equation*}
\end{proposition}

\begin{proof}
  Let $f\in\calC[\Gamma,k]$ be non null, and suppose w.l.o.g. that $f(0,\cdot)\neq 0$.
  Denote $t:=\lfloor\frac{d-1}{2}\rfloor$. Denote $V_0=\{0\}$ and, for $i\in[t]$, denote
  $V_i:=\left\{\sum_{s\in S_v}s\mid S_v\subseteq S, \Card{S}=i\right\}$
  the set of vertices at distance $i$ from $0$, and $V'_i:=\{v\in V_i\mid f(v,\cdot)\neq 0\}$.
  Remark that by Lemma \ref{lemma:generating-set-independent-subsets}, for $i\in[t]$ and $v\in V_i$, there is a unique set $S_v\subseteq S$ of size $\leq t$ such that $v=\sum_{s\in S_v}s$,
  because two distinct sets $S_v$ and $S'_v$ would create a linear dependancy of less than $d$ columns on $S$.

  For $i\in[t-1]$, by unicity of the decomposition, a vertex $v\in V'_i$ has $i$ neighbors in $V_{i-1}$ and no neighbors in $V_i$, therefore it has $n-i$ neighbors in $V_{i+1}$.
  Since $f(v,\cdot)\neq 0$, $v$ has at least $n-k+1$ non zero outgoing edges, $v$ has at least $n-k-i+1$ neighbors in $V'_{i+1}$.

  We prove by induction on $i=0,...,t$ that $\Card{V'_i}\geq\binom{n-k+1}{i}$.
  For $i=0$, $\Card{V'_i}=1$.
  Let $i\in[t]$.
  For $v\in V'_i$, denote $A_v:=E(v,[n])\cap V'_{i-1}$ the neighbors of $v$ in $V'_{i-1}$, and for $v'\in V'_{i-1}$, denote $B_{v'}:=E(v',[n])\cap V'_i$ the neighbors of $v'$ in $V'_i$, and denote $N(v,v'):=1$ if $v$ and $v'$ are neighbors, and $0$ otherwise.
  Then since any $v\in V'_i$ has at most $i$ neighbors in $V'_{i-1}$,
  \begin{equation}
    \label{eq:upper-bound-neighbors}
    \sum_{v\in V'_i}\sum_{v'\in V'_{i-1}}N(v,v')=\sum_{v\in V'_i}\Card{A_v}\leq i\Card{V'_i}\text,
  \end{equation}
  and since any $v'\in V'_{i-1}$ has at least $n-k-i+2$ neighbors in $V'_i$,
  \begin{equation}
    \label{eq:lower-bound-neighbors}
    \sum_{v'\in V'_{i-1}}\sum_{v\in V'_i}N(v,v')=\sum_{v'\in V'_{i-1}}\Card{B_{v'}}\geq(n-k-i+2)\Card{V'_{i-1}}\text.
  \end{equation}
  Combining (\ref{eq:upper-bound-neighbors}), (\ref{eq:lower-bound-neighbors}) and the induction, we obtain
  \begin{equation*}
    \Card{V'_i}
    %&\geq\frac{n-k-i+2}{i}\Card{V'_{i-1}}\\
    \geq \frac{n-k-i+2}{i}\binom{n-k+1}{i-1}
=\binom{n-k+1}{i}\text,
  \end{equation*}
  which concludes the induction.

  Therefore, because $t=\lfloor\frac{n-k+1}{2}\rfloor$, there are at least
  \[\sum_{i=0}^{t}\Card{V'_i}\geq\sum_{i=0}^{\lfloor\frac{n-k+1}{2}\rfloor}\binom{n-k+1}{i}=\vadisparaitre{\frac12\sum_{i=0}^{n-k+1}\binom{n-k+1}{i}=}2^{n-k}\]
  vertices corresponding to non zero local codewords. Hence $w_H(f)\geq\frac{n-k+1}{2}\cdot 2^{n-k}=(n-k+1)2^{d-3}$. Thus $\Delta_H(f)\geq 2^{d-r-2}\!\left(1-\frac{k-1}{n}\right)$.
\end{proof}

\begin{proposition}[Upper bound on the minimal distance]
  \label{proposition:upper-bound-minimal-distance}
  Assume that $n-k+1=d-1$.
  If $S$ is given by Lemma \ref{lemma:generating-set-independent-subsets}, with $\Gamma=\Cay(G,S)$, we have
  \begin{equation*}
    \label{eq:bad-upper-bound-minimal-distance}
    \Delta_H(\calC[\Gamma,k])\leq 2^{d-r-1}\left(1-\frac{k-1}{n}\right)\text.
  \end{equation*}
\end{proposition}

\begin{proof}
  Let $S'=\{s_1,...,s_{n-k+1}\}\subseteq S$.
  Let $L$ be the degree $k-1$ polynomial such that $L(x_{n-k+1})=1$ and for $\ell=n-k+2,...,n, L(x_\ell)=0$.
  Define $f\in W(\Gamma,\bbF)$ as follows.
  For $v\in G$ and $\ell\in[n]$, $f(v,\ell):=L(x_\ell)$ if $v\in\Span(S')$ and $f(v,\ell):=0$ otherwise.
  We prove that $\forall v\in G,\ell\in[n],f(v,\ell)=f(E(v,\ell),\ell)$. For $v\in G$ and $\ell\in[n]$, if $v,E(v,\ell)\in\Span(S')$, then $f(v,\ell)=f(E(v,\ell),\ell)$, and otherwise $f(v,\ell)=0=f(E(v,\ell),\ell)$. Therefore $f$ is well-defined.
  Moreover, since $\deg L=k-1$, $f\in \calC[\Gamma,k]$.
  By Lemma \ref{lemma:generating-set-independent-subsets}, since $n-k+1=d-1$, $\Span(S')$ has dimension $d-1$.
  Thus $w_H(f)=\frac12(n-k+1)2^{d-1}$ and $\Delta_H(\calC[\Gamma,k])\leq\frac{w_H(f)}{n2^{r-1}}=2^{d-r-1}\!\left(1-\frac{k-1}{n}\right)$.
\end{proof}

\noindent
Since there are no \nouveau{asymptotically good sequences of binary codes very close to be MDS}, $d$ will be asymptotically much smaller than $r$, and hence the minimal distance computed in Propositions \ref{proposition:lower-bound-minimal-distance} and \ref{proposition:upper-bound-minimal-distance} tends to zero when the length of the code tends to infinity.

\section{Conclusion}

This paper introduced a new IOPP protocol designed for codes on graphs. It achieves practical complexity, and soundness competing with used IOPP for Reed-Solomon codes.
Future research will focus on generalizing the cuts to more than two subsets and allowing multiple classes of equivalence of cut-graphs.

However, without arithmetization, it cannot provide new SNARK constructions.
Since our codes are built using Reed-Solomon codes as base codes, it may enable efficient arithmetizations.
Indeed, among the main arithmetization techniques, the PlonK variants \cite{GWZC19} and R1CS \cite{BCRSVW19} write the computation to be verified as an arithmetic circuit in which each gate represents a polynomial.
Then, writing that circuit as a De Bruijn graph \cite{Spi95} makes it regular and probably suitable for our protocol.

\section*{Acknowledgments}
We thank
Daniel Augot for his advices, guidance and proof-readings,
Jade Nardi and Christophe Levrat for the fruitful discussions that led to studying graph folding,
and Clément Chivet for his help with some mathematical tricks.
% And Pierre Loisel.

\bibliographystyle{alpha}
\bibliography{biblio}

\Hugo{Faut aussi vérifier que la biblio est bonne.}

\vfill

\end{document}